\documentclass[12pt,reqno]{amsart}
\usepackage{amsmath}
\usepackage{latexsym}
\usepackage{amsfonts}
\usepackage{amssymb}
\usepackage{color}
\usepackage{bbm,dsfont}
\usepackage{graphicx}

\newtheorem{proposition}{Proposition}

\newtheorem{lemma}{Lemma}
\newtheorem{corollary}{Corollary}

\theoremstyle{definition}

\newtheorem{example}{Example}

\newcommand{\real}{\mathbb R} 
\newcommand{\complex}{\mathbb C} 
\newcommand{\half}{\tfrac{1}{2}} 

\newcommand{\ip}[2]{\left\langle\,#1\,|\,#2\,\right\rangle} 
\newcommand{\ket}[1]{|#1\rangle} 
\newcommand{\no}[1]{\left\|#1\right\|} 

\newcommand{\mnr}{M_n(\real)}
\newcommand{\mnnr}{M_{2n}(\real)}
\newcommand{\mnnc}{M_{2n}(\complex)}
\newcommand{\Symp}{Sp(2n)}
\newcommand{\symp}{\mathfrak{sp}(2n)}
\newcommand{\SO}{SO(2n)}
\renewcommand{\exp}{\mathrm{exp}}
\newcommand{\diag}{\mathrm{diag}} 
\newcommand{\id}{\mathbbm{1}} 
\newcommand{\nul}{0} 

\newcommand{\gauss}{\mathcal{G}}
\newcommand{\gauge}{\mathcal{G}_{\sigma}}

\begin{document}

\title[The semigroup structure of Gaussian channels]{The semigroup structure of Gaussian channels}

\author[Heinosaari]{Teiko Heinosaari}
\address{Teiko Heinosaari, Niels Bohr Institute, Blegdamsvej 17, 2100 Copenhagen, Denmark}
\email{heinosaari@nbi.dk}

\author[Holevo]{Alexander S. Holevo}
\address{Alexander S. Holevo, Steklov Mathematical Institute, Gubkina 8, 119991 Moscow, Russia}
\email{holevo@mi.ras.ru}

\author[Wolf]{Michael M. Wolf}
\address{Michael M. Wolf, Niels Bohr Institute, Blegdamsvej 17, 2100 Copenhagen, Denmark}
\email{wolf@nbi.dk}

\begin{abstract} \ We investigate the semigroup structure of bosonic Gaussian quantum channels.
Particular focus lies on the sets of channels which are divisible,
idempotent or Markovian (in the sense of either belonging to
one-parameter semigroups or being infinitesimal divisible). We show
that the non-compactness of the set of Gaussian channels allows for
remarkable differences when comparing the semigroup structure with
that of finite dimensional quantum channels. For instance, every
irreversible Gaussian channel is shown to be divisible in spite of
the existence of Gaussian channels which are not infinitesimal
divisible. A simpler and known
consequence of non-compactness is the lack of generators for certain
reversible channels. Along the way we provide new representations
for classes of Gaussian channels: as matrix semigroup, complex
valued positive matrices or in terms of a simple form describing
almost all one-parameter semigroups.
\end{abstract}

\maketitle

\section{Introduction}\label{sec:intro}

A \emph{quantum channel} describes the input-output relation of a
quantum mechanical operation. Mathematically, it is described by a
completely positive map which is trace-preserving (in the
Schr\"odinger picture) or identity-preserving (in the Heisenberg
picture). Considering equal input and output spaces we can
concatenate quantum channels with themselves or other channels.
Clearly, such a concatenation is again a valid quantum operation, so
that the set of quantum channels forms a semigroup. Inverses only
belong to this semigroup if the channels describe unitary evolution.
Apart from this subgroup of reversible channels various subsets can
be distinguished according to their semigroup properties; for
instance (i) channels which are elements of one-parameter
semigroups, i.e., physically speaking solutions of time-independent
Markovian master equations, (ii) channels which are infinitesimal
divisible, e.g., solutions of time-dependent Markovian master
equations, (iii) channels
which are divisible into others in a non-trivial way, (iv) channels
which are indivisible, and (v) channels which are idempotent.

Whereas the characterization of one-parameter semigroups goes mainly
back to the seventies \cite{GoKoSu76, Lindblad76} the distinction of
the above sets has  been addressed more recently in \cite{WoCi08}
and methods to decide membership in (i) have been provided in
\cite{WoEiCuCi08}. The classical counterpart known as
imbedding problem for Markov chains was exhaustively studied in
probability theory (see \cite{Jo74} and the references therein.)

The characterization of the above sets seems to become vastly more
complex as the dimension $d$ of the considered system is increased:
whereas for qubits ($d=2$) essentially everything is known, the
problem of deciding membership in (i) turns out to be NP-hard with
increasing $d$ \cite{CuEiWo09}.

The present article is devoted to the study of the semigroup
properties of \emph{bosonic Gaussian channels} (i.e., `quasi-free'
maps)---a class where the underlying Hilbert space is infinite
dimensional. The restriction to Gaussian channels is motivated by
their physical relevance (they model, for instance, optical fibers
and occur all along with quadratic interactions) and it is suggested
by the mentioned complexity issue. Note that this restriction has
two flavors: when we ask whether a channel is in one of the above
sets we do not only restrict the channel under consideration, we
also restrict the involved one-parameter semigroups or
factorizations to within the Gaussian setting. In this way we can
escape from infinite dimensional Hilbert space into finite
dimensional phase space and formulate everything in terms of finite
dimensional matrix analysis.

 At this point one might expect that the basic picture of the set
of finite dimensional quantum channels carries over to the Gaussian
setting. There are, however, crucial differences, for instance:
whereas the set of quantum channels in finite dimensions is compact,
the set of Gaussian channels (even though having a finite
dimensional parametrization) is not; every irreversible Gaussian
channel turns out to be divisible in spite of the existence of
Gaussian channels which are not infinitesimal divisible etc.

The article is organized as follows: Sec.\ref{sec:channels}
introduces the basic notation and Sec.\ref{sec:product} shows that
the set of Gaussian channels is indeed a matrix semigroup. In
Sec.\ref{sec:reversible}, which is mainly provided for completeness,
we review the reversible case and emphasize the fact that not every
canonical transformation has a generating quadratic Hamiltonian.
Sec.\ref{sec:one-param} deals with one-parameter semigroups, for
which a simple representation is provided, and Sec.\ref{sec:infdiv}
has a closer look at infinitesimal divisible channels. In
Sec.\ref{sec:divisible} we prove that all irreversible Gaussian
channels are divisible by exploiting a simple mapping from the set
of Gaussian channels to the cone of complex positive matrices.
Finally, in Secs.\ref{sec:idempotent},\ref{sec:gauge} idempotent and
gauge-covariant channels are investigated before Sec.\ref{sec:open}
concludes with some open questions.

\section{Gaussian channels}\label{sec:channels}

Let $Q_j,P_j$, $j=1,\ldots,n$, be the canonical operators satisfying
the the canonical commutation relations ($\hbar=1$)
\begin{equation*}
[Q_j,P_k]=i\delta_{jk} \, ,  \quad [Q_j,Q_k]=[P_j,P_k]=0 \, .
\end{equation*}
We will use a notation $R=(Q_1,P_1,\ldots,Q_n,P_n)^T$ and for each
$\xi\in\mathbb{R}^{2n}$, we define $W_\xi=e^{i\xi^T\sigma R}$. Here
we have denoted
\begin{equation*}
\sigma \equiv \sigma_n = \oplus_{i=1}^n \sigma_1 \, , \quad \sigma_1
= \left(\begin{array}{cc}0 & 1 \\-1 & 0\end{array}\right) \, .
\end{equation*}
The unitary operators $W_\xi$ are called \emph{Weyl operators} and
they correspond to displacements in phase space.

A \emph{Gaussian channel} is a quantum channel which maps Gaussian
states into Gaussian states \cite{HoWe01,EiWo07,CaEiGiHo08}. The
mathematical form of Gaussian channels is best described in the
Heisenberg picture when we look at their action on the Weyl
operators. A Gaussian channel then corresponds to a mapping
\begin{equation*}
W_\xi \mapsto W_{X^T \xi} \ e^{-\half \xi^T Y\xi} \, ,
\end{equation*}
where $X,Y$ are real $2n\times 2n$-matrices. Here and
thereafter without loss of generality we restrict to channels which
map zero-mean states into zero-mean states. Complete positivity
(cp) imposes a constraint on the matrices $X,Y$, which can be
written as
\begin{equation}\label{eq:XYconstraint}
 Y \geq i\left(\sigma-X\sigma X^T\right) \, .
\end{equation}
From now on, we will identify the Gaussian channels with the pairs
$(X,Y)$ of real $2n\times 2n$-matrices satisfying
\eqref{eq:XYconstraint}. We denote by $\gauss$ the set of all
Gaussian channels.

It is often useful to depict a Gaussian channel by its action on the
first and second moments of a quantum state. We denote the first
moments by a vector $d\in\real^{2n}$ whose components are
expectation values $d_k\equiv\langle R_k\rangle$ and we define the
\emph{covariance matrix} as
\begin{equation*}
\Gamma_{kl}=\langle\{R_k-d_k,R_l-d_l\}_+\rangle \, .
\end{equation*}
 A Gaussian channel $(X,Y)$ then acts as
\begin{eqnarray}
d &\mapsto & X d \, ,\nonumber\\
\Gamma &\mapsto& X\Gamma X^T+Y\nonumber \, .
\end{eqnarray}

\begin{example}\label{ex:preparation}
The preparation of a Gaussian state is a simple instance of a
Gaussian channel. A Gaussian channel with $X=0$ has an
input-independent output state with covariance matrix $Y$. The
cp-condition \eqref{eq:XYconstraint} then reduces to $Y+i\sigma\geq
0$, which is nothing but the condition for $Y$ to be a valid
covariance matrix.
\end{example}

Generally speaking, the $Y$ contribution in a Gaussian channel
$(X,Y)$ can be regarded  as noise term. It follows from
\eqref{eq:XYconstraint} that $Y\geq\nul$. If $Y=\nul$, then the
cp-condition \eqref{eq:XYconstraint} implies that $X\sigma X^T =
\sigma$, meaning that $X$ is an element of the real symplectic group
$\Symp$. The group of unitary Gaussian channels is therefore
identified with $\Symp$. Let us notice, however, that generally $X$
can be any real matrix, as long as sufficient noise is added (i.e.,
$Y$ is large enough).

\section{Semigroup product}\label{sec:product}

Concatenation of two Gaussian channels is again a Gaussian channel.
 Hence, the set of Gaussian channels forms a semigroup. The semigroup product is given by
\begin{equation}\label{eq:product}
(X_1,Y_1) \cdot (X_2,Y_2) = (X_1 X_2, Y_1 + X_1 Y_2 X_1^T)
\end{equation}
and the identity element is $(\id,\nul)$.

Let us recall that we are using the Heisenberg picture and in the
Schr\"odinger picture the order of the product is opposite to the
order of application of channels. For instance, if $(X_1,Y_1)$ and
$(X_2,Y_2)$ describe optical fibers,  then the product $(X_1,Y_1)
\cdot (X_2,Y_2)$ corresponds to the channel in which  the signal
first goes through $(X_2,Y_2)$ and then through $(X_1,Y_1)$.

The semigroup product \eqref{eq:product} can also be written as an
ordinary matrix product. For each $X,Y\in\mnnr$, we denote
$x=X\otimes X$ and $y\in\real^{2n}\otimes\real^{2n}$ is the column
vector defined through the condition
\begin{equation}\label{eq:matvec}
\ip{e_i \otimes e_j}{y} = \ip{e_i}{Y e_j} \, .
\end{equation}
Then the mapping $\pi:\gauss\to M_{4n^2+2n+1}(\real)$ with
\begin{equation*}
 \pi(X,Y) = \left(\begin{array}{ccc} x & y & 0 \\ 0 & 1 & 0 \\ 0 & 0 & X \end{array}\right)
\end{equation*}
is an injective homomorphism. We conclude that the semigroup of
Gaussian channels is a matrix semigroup---a useful property when it
comes to the discussion of one-parameter semigroups.

\section{Reversible channels}\label{sec:reversible}

A channel is called \emph{reversible} if it has an inverse which is
also a channel (i.e., we mean physically reversible as opposed to
mathematically invertible). Reversible channels are exactly the
unitary channels, and a Gaussian channel $(X,Y)$ is reversible iff
$X\in\Symp$ and $Y=\nul$.

For completeness we briefly discuss the group structure of
reversible channels, which already exhibits some interesting
features. The sympletic group $\Symp$ is a non-compact connected Lie
group whose Lie algebra $\symp$ is given by all real matrices $s$
such that $(s\sigma)^T=s\sigma$. The exponential map
\begin{equation}
 \exp: \symp \to \Symp \, , \quad s\mapsto e^s
\end{equation}
is not surjective - a common feature of non-compact Lie groups.

In physical terms, the lack of surjectivity of the exponential map
means that there are canonical transformations $S\in\Symp$ for which
there is no `Hamiltonian matrix' $s\in\symp$ generating them via
$S=e^s$. Clearly, if we consider the corresponding unitary $U_S$
acting on Hilbert space (i.e., an element of the metaplectic
representation of $\Symp$ \cite{ArDuMuSi95}), then there is always a
Hamiltonian $\hat{H}$ such that $U_S=\exp i \hat{H}$. Such a
Hamiltonian may be obtained from the spectral decomposition of
$U_S$. The point is, however, that $\hat{H}$ cannot be a quadratic
expression in the canonical operators if $S\notin\exp(\symp)$.

A necessary condition for $S\in\exp(\symp)$ is that $S$ has a real
logarithm. We recall the following standard result from matrix
analysis \cite{Culver66}.

\begin{proposition}\label{prop:logarithm}
A real matrix $X\in\mnr$ has a real logarithm $L\in\mnr$ (i.e.
$X=e^L$) iff $X$ is non-singular and the Jordan blocks of $X$
corresponding to negative eigenvalues have even multiplicities.
\end{proposition}

The characterization of the set of $\exp(\symp)$ goes back to
Williamson \cite{Williamson39}. For our purposes, the following
partial characterization should suffice.

\begin{proposition}\label{prop:williamson}
Let $S\in\Symp$. If $-1$ is not an eigenvalue of $S$, then
$S\in\exp(\symp)$ iff $S$ has real logarithm.
\end{proposition}

A simple consequence of Prop. \ref{prop:williamson} is that if
$S\in\Symp$ is positive, then $S\in\exp(\symp)$.

An important subgroup $K(2n)$ of $\Symp$ is formed by those matrices
which are in addition orthogonal, i.e., $K(2n)=\Symp \cap \SO$. The
corresponding maps are called \emph{passive transformations} as they
preserve the number of particles \cite{ArDuMuSi95}. The subgroup
$K(2n)$ is a (maximal) compact subgroup of $\Symp$ and it is
isomorphic to $U(n)$. Consequently, the exponential map from the Lie
algebra $\mathfrak{K}(2n)$ to $K(2n)$ is surjective. That is, every
$S\in K(2n)$ has the property that
$S\in\exp(\mathfrak{K}(2n))\subset\exp(\symp)$.

Using the Euler decomposition (cf. \cite{ArDuMuSi95}) each
$S\in\Symp$ can be written as a product $S=K_1DK_2$, where
$K_1,K_2\in K(2n)$ and $D\in\Symp$ is a diagonal matrix of the form
$D=\textrm{diag}(d_1,1/d_1,\ldots,d_n,1/d_n)$ with $d_1,\ldots,d_n >
0$.
The diagonal matrix $D$ describes single-mode squeezings.
Thus, every reversible Gaussian channel is a concatenation of passive transformations and single-mode
squeezings.

From the previous discussion we arrive at the following conclusion.

\begin{proposition}\label{prop:SS}
Let $S\in\Symp$. There are $S_1,S_2\in\exp(\symp)$ such that
$S=S_1S_2$.
\end{proposition}

\begin{proof}
Since $S$ has an Euler decomposition
\begin{equation*}
S=K_1DK_2=K_1DK_1^TK_1^{-T}K_2 \, .
\end{equation*}
it is a product of a positive  $S_1:=K_1DK_1^T$ and an orthogonal
$S_2=K_1^{-T}K_2$ symplectic matrix. Both $S_1,S_2\in\exp(\symp)$,
as discussed earlier.
\end{proof}

\section{One-parameter Gaussian semigroups}\label{sec:one-param}

\subsection{General form}

By \emph{ one-parameter  semigroup of Gaussian channels} we mean a
family of Gaussian channels, parametrized by $\real^+$, which
satisfies the following conditions:
\begin{itemize}
\item[(i)] \emph{continuity}: $X_t$ and $Y_t$ depend on $t\in\real^+$ in a continuous way
\item[(ii)] \emph{semigroup property}: $X_tX_s=X_{t+s}$ and $Y_{t+s}=Y_t + X_tY_sX_t^T$.
\item[(iii)] \emph{connected to the identity}: $X_0=\id$ and $Y_0=\nul$
\end{itemize}

As we have seen in Section \ref{sec:product}, the semigroup of
Gaussian channels is a matrix semigroup. It follows that every
one-parameter semigroup of Gaussian channels has a generator
\cite{SCOS06}. In particular, the mappings $t\mapsto X_t$ and
$t\mapsto Y_t$ are differentiable. The following characterization
(although in a slightly different form) has been derived in
\cite{Vanheuverzwijn78}.

\begin{proposition}\label{prop:one-param-gen}
 A family $(X_t,Y_t)_{t\geq 0}$ of Gaussian channels
 forms a one-parameter semigroup iff there exists real
 matrices $A,B,H$ with $iA+B\geq \nul$, $A^T=-A$, $H^T=H$ such that
 \begin{eqnarray}
  X_t &=& e^{t(A-H)\sigma} \, , \label{eq:one-param-gen-X} \\
  Y_t &=& 2\int_0^t X_{s}^T \sigma^T B \sigma X_s \ ds \, . \label{eq:one-param-gen-Y}
 \end{eqnarray}

In the Heisenberg picture this corresponds to an evolution of any
observable $O$ governed by a master equation $\partial_t
O=\mathcal{L}(O)$ with a Liouvillian
\begin{eqnarray}
\mathcal{L}(O) &=& i[\hat{H},O]+\sum_\alpha \hat{L}_\alpha^\ast O
\hat{L}_\alpha-\frac12\{\hat{L}_\alpha^\ast\hat{L}_\alpha,O\}_+ \, , \label{eq:Liouvillian} \\
\hat{H} &=& \frac12 \sum_{kl} H_{kl} R_kR_l \, ,\\
\hat{L}_\alpha &=&\sum_k L_{\alpha,k}R_k,\quad\mbox{with\quad}
B+iA=L^\ast L \label{eq:L} \, .
\end{eqnarray}
\end{proposition}

A simple consequence of this characterization is the following.

\begin{corollary}\label{cor:Xlog}
Let $X\in\mnnr$. The following conditions are equivalent:
\begin{itemize}
\item[(i)] There exists a matrix
 $Y\in\mnnr$ such that $(X,Y)$ is an element of a
 one-parameter semigroup of Gaussian channels.
\item[(ii)] $X$ is non-singular and the Jordan blocks of $X$
corresponding to negative eigenvalues have even multiplicities.
\end{itemize}
\end{corollary}

\begin{proof}
By Prop.\ref{prop:logarithm} the condition (ii) on $X$ is nothing
but the existence of a real logarithm $L$ which is clearly necessary
for $X$ to occur in an element of a one-parameter semigroup of
Gaussian channels. Hence, (i) implies (ii).

In order to see the other direction, suppose that (ii) holds and let
$L$ be a real logarithm of $X$. Let us decompose $L\sigma^T=A-H$
into a symmetric part $H=H^T$ and anti-symmetric part $A=-A^T$,
respectively. Then there is always a $B\in\mnnr$ (e.g. $B =
\no{A}_\infty\id$) such that $iA+B\geq 0$ and we can construct a
one-parameter semigroup of Gaussian channels by following the
characterization in Prop.\ref{prop:one-param-gen}.
\end{proof}

Prop. \ref{prop:one-param-gen} gives a complete but rather
cumbersome characterization of one-parameter semigroups of Gaussian
channels. In particular, the appearing integral might complicate
further use of the characterization. Fortunately, almost all
generators of such semigroups allow for a simpler representation
discussed in the next subsection.

\subsection{Simple form}

Suppose that $\{X_t\}_{t\geq 0}$ is a semigroup and fix a real symmetric matrix $\mathcal{Y}$.
Then by setting
 \begin{equation}\label{eq:one-param-simple}
  Y_t = \mathcal{Y}-X_t \mathcal{Y} X_t^T
 \end{equation}
the semigroup property for $(X_t,Y_t)_{t\geq 0}$ is satisfied. The
cp-condition now reads
\begin{equation}\label{eq:cp-simple}
  \mathcal{Y}- i\sigma \geq  X_t (\mathcal{Y} - i\sigma) X_t^T \, .
\end{equation}

\begin{example}
An \emph{amplification channel} is of the form $X=\sqrt{\eta}\ \id$,
$Y=(\eta-1)\ \id$ for some $\eta\in (1,\infty)$.
Amplification channels form a one-parameter semigroup. Namely,
\begin{equation*}
X_t=e^{t}\id \, , \quad Y_t=(e^{2t}-1)\id \, .
\end{equation*}
This is of the simple form \eqref{eq:one-param-simple} with
$\mathcal{Y}=-\id$.
\end{example}

\begin{proposition}\label{prop:one-param-simple}
A one-parameter Gaussian semigroup $(X_t,Y_t)_{t\geq 0}$ is of the
simple form \eqref{eq:one-param-simple} if the operator
$(A-H)\sigma$ in Prop. \ref{prop:one-param-gen} (i.e., the generator
of $X_t$) does not have a pair of eigenvalues of the form
$\pm\lambda$.
\end{proposition}

\begin{proof}
The general form \eqref{eq:one-param-gen-Y} gives $\dot{Y}_{0} = 2\
\sigma^T B \sigma\equiv \tilde{B}$. On the other hand, the simple
form \eqref{eq:one-param-simple} leads to
\begin{equation*}
\dot{Y}_{0} = \tilde{A}\mathcal{Y} + \mathcal{Y}\tilde{A}^T \, ,
\end{equation*}
where $\tilde{A}=(H-A)\sigma$. Hence, in order for $(X_t,Y_t)_{t\geq
0}$ to be of the form \eqref{eq:one-param-simple}, we need to find
$\mathcal{Y}$ such that
\begin{equation}\label{eq:condition-for-simple-form}
 \tilde{A}\mathcal{Y} + \mathcal{Y}\tilde{A}^T = \tilde{B}\, .
\end{equation}
The linear equation \eqref{eq:condition-for-simple-form} can be
written as
\begin{equation*}
 ( \id \otimes \tilde{A}^T + \tilde{A} \otimes \id ) \ket{\mathcal{Y}} = \ket{\tilde{B}}
\end{equation*}
Hence, if $\id \otimes \tilde{A}^T + \tilde{A} \otimes \id$ is
invertible, then we have a solution for $\mathcal{Y}$.

The eigenvalues of $\id \otimes \tilde{A}^T + \tilde{A} \otimes \id$
are of the form $\lambda_i + \lambda_j$, where $\lambda_i,\lambda_j$
are eigenvalues of $\tilde{A}$ (this can be seen e.g. using Schur
upper-triangular form for $\tilde{A}$ and $\tilde{A}^T$). Therefore,
if $\tilde{A}$ has the property that the sum of any two of its
eigenvalues is nonzero, the invertibility of $\id \otimes
\tilde{A}^T + \tilde{A} \otimes \id$ follows.

If $\mathcal{Y}$ is a solution then $\mathcal{Y}^T$ and hence
$\frac{1}{2}\left(\mathcal{Y}+\mathcal{Y}^T\right)$ are again
solutions. Therefore, $\mathcal{Y}$ can be chosen symmetric.
\end{proof}

Prop.\ref{prop:one-param-simple} implies that almost all
one-parameter semigroups admit a representation of the simple form
in Eq.\eqref{eq:one-param-simple}. Moreover, it shows that we can
approximate any one-parameter Gaussian semigroup with a
one-parameter semigroup of the simple form. Namely, assume
$(X_t,Y_t)_{t\geq 0}$ is a one-parameter Gaussian semigroup which is
not of the simple form \eqref{eq:one-param-simple}. Suppose that
$A,B,H$ are as in Prop. \ref{prop:one-param-gen}. We can clearly
make an arbitrarily small change in the matrices $A$ and $H$ (hence
getting new matrices $A'$ and $H'$) in a way that the set of
eigenvalues of $(A'-H')\sigma$ do not contain pairs of the type
$\pm\lambda$. If necessary, we also make a small change to $B$,
obtaining $B'$, to guarantee the condition $iA'+B'\geq\nul$. Hence,
there exists a one-parameter semigroup  $(X'_t,Y'_t)_{t\geq 0}$ of
the simple form which has generating matrices  $A',B',H'$ arbitrary
close to $A,B,H$.

\begin{example}
Not all one-parameter Gaussian semigroups are of the simple form
\eqref{eq:one-param-simple}. For instance, suppose $n=1$ and choose
$A=\nul$ together with
\begin{equation*}
H=\left( \begin{array}{cc}0 & 1 \\1 & 0\end{array} \right) \, .
\end{equation*}
In this case
\begin{equation*}
\tilde{A}=(H-A)\sigma=\left( \begin{array}{cc} -1 & 0 \\0 &
1\end{array} \right) \, .
\end{equation*}
The generated matrix $X_t=e^{-t\tilde{A}}$ corresponds to squeezing.

With any $\mathcal{Y}$, the matrix
$\tilde{A}\mathcal{Y}+\mathcal{Y}\tilde{A}$ is diagonal. Therefore,
the condition \eqref{eq:condition-for-simple-form} cannot be
satisfied whenever $B$ (and hence $\tilde{B}$) is a non-diagonal
matrix. On the other hand, since $A=\nul$ any positive matrix is a
possible choice for $B$.
\end{example}

\subsection{Bounded evolutions}

We say that a one-parameter semigroup $(X_t,Y_t)_{t\geq 0}$ has
\emph{bounded noise term} if there is a constant $c$ such that
$\no{Y_t}\leq c$ for all $t\in\real^+$.

\begin{proposition}
Let $(X_t,Y_t)_{t\geq 0}$ be a one-parameter semigroup of Gaussian
channels. The following conditions are equivalent:
\begin{itemize}
 \item[(i)] There exists a positive matrix $\mathcal{Y}$ such that
 \begin{equation}\label{eq:YY}
  Y_t = \mathcal{Y}-X_t \mathcal{Y} X_t^T \, .
 \end{equation}
 \item[(ii)] $(X_t,Y_t)_{t\geq 0}$ has bounded noise term.
\end{itemize}
\end{proposition}

\begin{proof}
Suppose that (i) holds. For every $t\in\real^+$, we then have
 \begin{equation}
  \nul \leq Y_t =\mathcal{Y}-X_t\mathcal{Y}X_t^T \leq \mathcal{Y} \, .
 \end{equation}

Suppose that (ii) holds. Let $\mu$ be an invariant mean of the
semigroup $\real^+$ (see e.g. \cite{AHA1}). By the assumption, each
matrix entry $t\mapsto [Y_t]_{ij}$ is a bounded continuous function.
We define the matrix $\mathcal{Y}$ by
\begin{equation*}
[\mathcal{Y}]_{ij} = \mu([Y_\cdot]_{ij}) \, .
\end{equation*}
As $\mu$ is invariant, an application to the second semigroup
condition $Y_{t+s}=Y_t + X_tY_sX_t^T$ gives the formula
\eqref{eq:YY}. Moreover, $\mathcal{Y}\geq\nul$ since $Y_t\geq\nul$
for each $t\in\real^+$ and
\begin{equation*}
 \ip{v}{\mathcal{Y}v}=\mu \left( \ip{v}{Y_{\cdot}v} \right) \geq 0 \, .
\end{equation*}
\end{proof}

Let $(X_t,Y_t)_{t\geq 0}$ be a one-parameter semigroup and suppose
there is a $\mathcal{Y}$ satisfying \eqref{eq:one-param-simple}. The
matrix $\mathcal{Y}$ is a valid covariance matrix (e.g., of a
Gaussian state) iff
\begin{equation}\label{eq:covariance}
\mathcal{Y}-i\sigma \geq 0 \, .
\end{equation}
Thus, in this situation the one-parameter semigroup
$(X_t,Y_t)_{t\geq 0}$ has an invariant Gaussian state.

This is the case if the semigroup $\left( X_{t}\right)
_{t\geq 0}$ is strictly contractive: $\left\Vert X_{t}\right\Vert
<1$ for $t>0.$ Then we have $\lim_{t\to\infty} X_t =0$ so that \eqref{eq:cp-simple} becomes \eqref{eq:covariance}.

\begin{example}
An \emph{attenuation channel} is of the form $X=\sqrt{\eta}\ \id$,
$Y=(1-\eta)\ \id$  for some $\eta\in (0,1)$.
Attenuation channels form a bounded one-parameter semigroup. Namely,
\begin{equation*}
X_t=e^{-t}\id \, , \quad Y_t=(1-e^{-2t})\id \, .
\end{equation*}
For this one-parameter semigroup we have $\mathcal{Y}=\id$.
Therefore, the  vacuum state (with covariance matrix $\Gamma=\id$
and first moments $d=0$) is an invariant state for the one-parameter
semigroup of attenuation channels.
\end{example}

We notice that a one-parameter semigroup $(X_t,Y_t)_{t\geq 0}$ may
have bounded noise term without having an invariant Gaussian state.
For a simple example suppose that $(X_t,Y_t)_{t\geq 0}$ is
a one-parameter semigroup of reversible channels. Then
$X_t=e^{-tH\sigma}$ for some symmetric matrix $H$ and $Y_t=0$. If
$X_t$ is not orthogonal (in which case $\mathcal{Y}\propto\id$) we
have $\mathcal{Y}=0$, which is clearly not a valid covariance
matrix.

\section{Infinitesimal divisible channels}\label{sec:infdiv}

 We call a Gaussian channel parameterized by $(X,Y)$
\emph{infinitesimal divisible} if either
\begin{itemize}
\item[(a)] for every $\varepsilon >0$ there exists a finite set of Gaussian
channels $(X_i,Y_i)$ such that:
\begin{itemize}
\item[(i)] $\no{(X_i,Y_i)-(\id,\nul)}<\varepsilon$
\item[(ii)] $\prod_i (X_i,Y_i)=(X,Y)$
\end{itemize}
or
\item[(b)] the channel can be approximated arbitrarily well with (a)-type of channels.
\end{itemize}

We note that in the classical case  \cite{Jo74} this is what is called limit of
a null triangular array, and it is proved there that such limits are
precisely solutions of time-dependent Kolmogorov equations.
Also in the quantum case solutions of time-dependent Markovian master equations are clearly infinitesimal divisible. For finite dimensional quantum systems the close relation between the two sets has been studied in \cite{WoCi08}.

The concatenation of two infinitesimal divisible channels is clearly
infinitesimal divisible. The infinitesimal divisible channels thus
form a subsemigroup of $\gauss$.

It is clear that if  $(X,Y)$ is an element of a one-parameter
semigroup of Gaussian channels, then it is infinitesimal divisible.
The converse is, however, not true. Namely, suppose that $S$ is a
symplectic matrix such that $S\notin \exp(\symp)$. Then the channel
$(S,0)$ is not an element of a one-parameter Gaussian semigroup, but
by Prop. \ref{prop:SS} we have symplectic matrices $S_1,S_2$ such
that $S=S_1S_2$ and $S_1,S_2\in\exp(\symp)$.
 Since the channels $(S_1,0)$ and $(S_2,0)$ are infinitesimal divisible,
 so is also $(S,0)$.

The following is a simple necessary condition for a channel to be
infinitesimal divisible.

\begin{proposition}\label{prop:infdiv}
If a Gaussian channel $(X,Y)$ is infinitesimal divisible, then $\det
X \geq 0$.
\end{proposition}

\begin{proof}
The claim follows from the continuity and multiplicativity of the
determinant.
\end{proof}

\begin{example}
The Gaussian channel describing a phase conjugating mirror (with
minimal noise) corresponds to the choices
\begin{equation*}
X= \id_n \otimes \left(\begin{array}{cc}1 & 0 \\0 &
-1\end{array}\right) \, ,  \quad Y=2 \id_{2n} \, .
\end{equation*}
If $n$ is odd, then $\det(X)<0$ and therefore $(X,Y)$ is not
infinitesimal divisible.
\end{example}

In the spirit of Corollary \ref{cor:Xlog} we can formulate a
converse of Prop.\ref{prop:infdiv}.

\begin{proposition}\label{prop:weakconv}
If $X\in\mnnr$ satisfies $\det(X)>0$, then there exists  a
$Y\in\mnnr$ such that $(X,Y)$ is an infinitesimal divisible Gaussian
channel.
\end{proposition}

\begin{proof}
Consider the real Jordan decomposition $X=MJM^{-1}$ and group the
Jordan  blocks in $J$ such that $J=J_-\oplus J_r$ with $J_-$ being
the collection of all Jordan blocks with negative eigenvalues and
$J_r$ containing all the others. By defining $X_1\equiv
M(-\id)\oplus\id M^{-1}$ and $X_2\equiv M(-J_-)\oplus J_r M^{-1}$ we
get $X=X_1X_2$. Since $\det(X)>0$, the multiplicity of the
eigenvalue $-1$ of $X_1$ is even. By Proposition 1, the $X_i$'s
now have real logarithms which implies by Corollary \ref{cor:Xlog}
that there are $Y_i$'s such that $(X_i,Y_i)$ are elements of a
one-parameter semigroup of Gaussian channels. Consequently, $(X,Y)$
is infinitesimal divisible if we choose $Y=Y_1+X_1Y_2X_1^T$.
\end{proof}

We also have the following simple observation.

\begin{proposition}
Let $(X,Y)$ be infinitesimal divisible Gaussian channel. Then every
Gaussian channel $(X,\widetilde{Y})$ with $\widetilde{Y}\geq Y$ is
infinitesimal divisible.
\end{proposition}

\begin{proof}
We can split the additional noise $\widetilde{Y}-Y$ into arbitrarily
small pieces, $\prod_{j=1}^m (\id, \frac{1}{m} (\widetilde{Y}-Y)) =
(\id, \widetilde{Y}-Y)$.
\end{proof}

\section{Divisible channels}\label{sec:divisible}

 We call a Gaussian channel $(X,Y)$ \emph{divisible} if there exist two non-reversible
 Gaussian channels $(X_1,Y_1)$ and $(X_2,Y_2)$ such that
\begin{equation*}
(X_1,Y_1)\cdot(X_2,Y_2)=(X,Y) \, .
\end{equation*}
In this section we show that actually all non-reversible Gaussian
channels are divisible---in spite of the existence of Gaussian channels which are not infinitesimal divisible. This is in sharp contrast to the finite
dimensional case, where one has also indivisible channels
\cite{WoCi08}.

The main tool which we use to prove the result is a surjective
mapping from $\gauss$ onto the cone of positive matrices in $\mnnc$.
Thus we define the mapping $p$ in the following way:
\begin{equation*}
p: \mnnr \times \mnnr \to \mnnc \, , \quad (X,Y) \mapsto i (X \sigma
X^T - \sigma) + Y\, .
\end{equation*}
A comparison of this definition with the cp-condition
\eqref{eq:XYconstraint} shows  that $p(X,Y) \geq \nul$ iff
$(X,Y)\in\gauss$. Moreover, $p(X,Y)=\nul$ iff $X\in\Symp$ and
$Y=\nul$. Thus, the kernel of $p$ is exactly the set of reversible
elements of $\gauss$.

An essential property of $p$ is the fact each positive matrix $P$ is
an image of a  Gaussian channel $(X,Y)$. This property is proved in
the following two lemmas. This first lemma is a standard result in linear algebra, but we give a proof for the reader's convenience.

\begin{lemma}\label{lemma:anti}
Suppose $M\in\mnnr$ is anti-symmetric. Then it can be written in the
form $M=N \sigma N^T$ for some $N\in\mnnr$.
\end{lemma}

\begin{proof}
 Since $M$ is
normal, there is an orthogonal real matrix $R$ such that $R^TMR$ is
a block diagonal matrix where each block is a 1-by-1 matrix or a
2-by-2 matrix of the form
\begin{equation}\label{eq:2-by-2}
\left(\begin{array}{cc}a & b \\-b & a\end{array}\right) \, , \quad
a,b\in\real, b >0 \, .
\end{equation}
As $M$ is anti-symmetric every 1-by-1 block has to be $0$ and every
2-by-2 block of the form \eqref{eq:2-by-2} has $a=0$. For this kind
of 2-by-2 matrix we can write
\begin{equation*}
\left(\begin{array}{cc}0 & b \\-b & 0\end{array}\right) =
\left(\begin{array}{cc}\sqrt{b} & 0 \\ 0 &
\sqrt{b}\end{array}\right) \left(\begin{array}{cc} 0 & 1 \\-1 &
0\end{array}\right) \left(\begin{array}{cc} \sqrt{b} & 0 \\ 0 &
\sqrt{b} \end{array}\right) \, .
\end{equation*}
\end{proof}

\begin{lemma}
Let $P\in\mnnc$ and $P\geq\nul$. There is $(X,Y)\in\gauss$ such that
$p(X,Y)=P$.
\end{lemma}

\begin{proof}
Let us write $P=P_R + i P_I$, where $P_R,P_I\in\mnnr$.  Since $P\geq
\nul$, $P_R$ is symmetric positive and $P_I$ is
anti-symmetric. Also the matrix $P_I + \sigma$ is anti-symmetric,
and by Lemma \ref{lemma:anti} it can thus be written in the form
\begin{equation*}
P_I + \sigma = X \sigma X^T \, ,
\end{equation*}
where $X\in\mnnr$. Choosing $Y=P_R$ the claim is proven.
\end{proof}

Each positive matrix $P$ represents an equivalence class of Gaussian
channels rather than a single channel. Namely, for two Gaussian
channels $(X_1,Y_2)$ and $(X_2,Y_2)$ we have
\begin{equation*}
p(X_1,Y_1)=p(X_2,Y_2) \qquad\textrm{iff}\qquad Y_1=Y_2 \quad
\textrm{and}\quad X_1\sigma X_1^T = X_2 \sigma X_2^T \, .
\end{equation*}
If we concatenate two Gaussian channels corresponding to positive
matrices  $P_1$ and $P_2$ we obtain a Gaussian channel in the
equivalence class
\begin{equation}
P = P_1 + X_1 P_2 X_1^T \, .
\end{equation}

With this preparation we are now ready to prove that every
non-reversible Gaussian  channel is divisible.

\begin{proposition}
Let $(X,Y)\in\gauss$ be non-reversible. There exist non-reversible
$(X_1,Y_1),(X_2,Y_2)\in\gauss$ such that $(X,Y)=(X_1,Y_1)\cdot
(X_2,Y_2)$.
\end{proposition}

\begin{proof}
Let us first consider the case $\det X = 0$. Choose $X_1=X$,
$Y_1=Y$, $X_2=\id$ and $Y_2$ the projector onto the kernel of $X$.
In this way we can write the channel $(X,Y)$ as a concatenation of
two non-reversible channels.

Let us  then
suppose $\det X \neq 0$. Let $P=p(X,Y)$. We denote $P_1=\varepsilon
P$, and fix a pair $(X_1,Y_1)$ such that $p(X_1,Y_1)=P_1$. It is
possible to choose $0<\varepsilon <1$ such that $\det X_1 \neq 0$.
Indeed, by its definition $X_1$ satisfies
\begin{equation}
 X_1 \sigma X_1^T = \varepsilon X \sigma X^T + (1-\varepsilon) \sigma \, ,
\end{equation}
hence we get
\begin{equation*}
(\det X_1)^2 = \det (X_1 \sigma X_1^T) =  \det( \varepsilon X \sigma
X^T + (1-\varepsilon) \sigma ) \, .
\end{equation*}
For $\varepsilon=0$ the right hand side is 1, and from the
continuity  of the determinant follows that $\det X_1 \neq 0$ for
some $0<\varepsilon <1$.

We then define
\begin{equation*}
P_2 := (1-\varepsilon) X_1^{-1}PX_1^{-T} \, .
\end{equation*}
It follows that
\begin{equation*}
X_1 P_2 X_1^T + P_1 = P \, .
\end{equation*}
This shows that in the equivalence class of Gaussian channels
represented by  $P$, there is at least one channel $(\tilde{X},Y)$
which can be divided non-trivially. So it remains to prove that this
holds then for all Gaussian channels in the equivalence class.

As we have noticed earlier, $(X,Y)$ and $(\tilde{X},Y)$ are in the
same equivalence class  iff $X\sigma X^T = \tilde{X} \sigma
\tilde{X}^T$. Thus, $S:= X^{-1}\tilde{X} \in\Symp$ and
$X=\tilde{X}S^{-1}$. Therefore, a decomposition for $(\tilde{X},Y)$
also leads to a decomposition of $(X,Y)$.

Since $Y_1=\varepsilon Y\neq\nul$, the channel $(X_1,Y_1)$ is
non-reversible. Also the channel $(X_2,Y_2)$ has to be
non-reversible as $Y_1 \neq Y$ implies  that $Y_2\neq\nul$.
\end{proof}

Note that $(X_1,Y_1)$ can be chosen arbitrary close to the ideal channel $(\id,0)$. That is, form an arbitrary irreversible Gaussian channel we can `chop off' an infinitesimal (irreversible) piece so that the remaining part is still a valid Gaussian quantum channel. The possibility of iterating this procedure, i.e., chopping off an infinitesimal pieces from the remainder and so further, might suggest that every Gaussian channel is infinitesimal divisible. This intuition, however, fails since the remaining channel (after having chopped off a piece) is, in the Gaussian context, not necessarily closer to the ideal channel (as it would be in the finite dimensional context). In fact, for Gaussian channel which are not infinitesimal divisible this procedure can bring us further and further away from the identity. This might be seen as a signature of the non-compactness of the set of Gaussian channels (as opposed to the compactness in the finite-dimensional context).

\section{Idempotent channels}\label{sec:idempotent}

A Gaussian channel $(X,Y)$ is \emph{idempotent} if
\begin{equation*}
(X,Y)\cdot(X,Y)=(X,Y) \, .
\end{equation*}
This leads to the requirements $X^2=X$ and $XYX^T=0$. Since $Y$ is
positive, the second condition can be written in the form
$(X\sqrt{Y})(X\sqrt{Y})^T=0$, which is equivalent to $XY=0$.
Therefore, we conclude that $(X,Y)$ is idempotent iff
\begin{equation}\label{eq:idempotent-conditions}
X^2=X \, , \quad XY=0 \, .
\end{equation}
In physical terms, idempotency means that a repeated use of the
channel does not  change the system any further.

Suppose that $X$ is invertible. Then the conditions
\eqref{eq:idempotent-conditions} imply that $(X,Y)=(\id,\nul)$,
which is just the identity channel. However, there are also other
idempotent channels as illustrated in the following example.

\begin{example}\label{ex:idempotent}
Let $X$ and $Y$ be diagonal matrices of the form
\begin{equation*}
X=\diag(\overset{2k}{\overbrace{1,\ldots,1}},\overset{2n-2k}{\overbrace{0,\ldots,0}})
\, , \quad Y=\diag(\overset{2k}{ \overbrace{0,\ldots
,0}},\overset{2n-2k}{\overbrace{y_{1},y_{1},\ldots,y_{n-k},y_{n-k}}})
\, ,
\end{equation*}
where $y_{j}\geq 1;j=1,\dots ,n-k$. Then the pair $(X,Y)$ clearly
satisfies \eqref{eq:idempotent-conditions}.  The cp-condition
\eqref{eq:XYconstraint} breaks into conditions for $2\times
2$-matrices,
\begin{equation*}
i \sigma_1 \geq i\sigma_1 \, , \quad y_j \id \geq i\sigma_1 \, ,
\end{equation*}
which obviously hold. The channel $(X,Y)$ corresponds to a
transformation where we do nothing for the  first $2k$ modes but for
the rest $2n-2k$ modes we do a state preparation (see Example
\ref{ex:preparation}).
\end{example}

Let $S\in\Symp$. If we concatenate a channel $(X,Y)$ with the
reversible channels corresponding to $S$ and $S^{-1}$, we get
\begin{equation}
(S,0) \cdot (X,Y) \cdot (S^{-1},0) = (SXS^{-1},SYS^T) \, .
\end{equation}
It is easy to verify that if $(X,Y)$ is idempotent, then also
$(SXS^{-1},SYS^T)$ is idempotent. Therefore, Example
\ref{ex:idempotent} generates a full class of idempotent channels.
In the following we show that actually all idempotent channels are
of that form.

\begin{proposition}
A Gaussian channel $(X,Y)$ is idempotent iff there is a symplectic
matrix $S$ such that
\begin{equation}\label{eq:x-diag}
SXS^{-1}=\diag(\overset{2k}{\overbrace{1,\ldots
,1}},\overset{2n-2k}{ \overbrace{0,\ldots ,0}}) \, ,
\end{equation}
\begin{equation}\label{eq:can}
SYS^T=\diag(\overset{2k}{ \overbrace{0,\ldots
,0}},\overset{2n-2k}{\overbrace{y_{1},y_{1},\ldots
,y_{n-k},y_{n-k}}})\,,
\end{equation}
where $y_{j}\geq 1;j=1,\dots ,n-k$.
\end{proposition}

\begin{proof}
Let us first show that $X^T$ is a symplectic projection, i.e.  the
symplectic space $\left( \mathbb{R}^{2n},\sigma \right) $ is a
direct sum of two subspaces $ V_{1},V_{2}$, mutually orthogonal with
respect to the symplectic form. That is every vector $v$ is uniquely
decomposed as $v_{1}+v_{2}$ with  $v_{1}\in V_{1},v_{2}\in V_{2},$
and $v_{1}^{T}\sigma v_{2}=0.$

Indeed, put $v_{1}=X^Tv,v_{2}=(\id -X^T)v,$ then $v=v_{1}+v_{2}$ and
$Yv_{1}=0$. By condition \eqref{eq:XYconstraint} we get
\begin{equation*}
v_{2}^{T}Yv_{2}=(v_{1}+iv_{2})^{\ast }Y(v_{1}+iv_{2})\geq
2v_{1}^{T}\sigma v_{2}
\end{equation*}
for all $v_{1}\in V_{1},v_{2}\in V_{2},$ hence $v_{1}^{T}\sigma
v_{2}=0$.

From (\ref {eq:XYconstraint}) we also obtain
\begin{equation*}
Y=(\id-X)Y(\id-X)^T\geq i(\id-X)\sigma (\id-X)^T \, .
\end{equation*}
Applying symplectic diagonalization of the positive symmetric matrix
$ (\id-X)Y(\id-X)^T$ in $V_{2}$, one can always find a symplectic
matrix $S$ in $\left( R^{2n},\sigma \right) $ satisfying
(\ref{eq:can}).

Since $XY=\nul$, we also have $SXS^{-1}SYS^T=\nul$ and
\eqref{eq:x-diag} follows.
\end{proof}

\section{Gauge-covariant channels}\label{sec:gauge}

We say that a Gaussian channel $(X,Y)$ is \emph{gauge-covariant} if
$[X,\sigma]=[Y,\sigma]=0$.  We denote by $\gauge$ the set of all
gauge-covariant Gaussian channels. It is clearly a subsemigroup of
$\gauss$. Physically, gauge-covariant channels arise for instance
from a number conserving (i.e., passive) coupling to an environment
(cf. \cite{Wolf08}).

Let us rearrange the matrix $\sigma $ such that
\[
\sigma =\left(
\begin{array}{cc}
0 & \id_n \\
-\id_n & 0
\end{array}
\right) ,
\]
then matrices $M\in M_{2n}(\mathbb{C})$ commuting with $\sigma $ are
those of the form
\begin{equation}
M=\left(
\begin{array}{cc}
A & B \\
-B & A
\end{array}
\right) ,  \label{Eq:M}
\end{equation}

where $A,B\in M_{n}(\mathbb{C})$.
The maps $M\mapsto A\pm iB$ are easily seen to be $\ast -$homomorphisms of $M_{2n}(\mathbb{C})$
onto $M_{n}(\mathbb{ C})$, hence $M\geq 0$ implies $A\pm iB\geq 0$.
Note that $A^{\ast }=A,B^{\ast }=-B$ in this case. Let us show that
conversely, $M\geq 0$ if $A\pm iB\geq 0. $ Let $A\pm iB\geq 0,$ then
$A\geq 0$ and $B^{\ast }=-B.$ Assume first that $A $ is
nondegenerate, then $\id_n\pm iA^{-1/2}BA^{-1/2}\geq 0,$ which
implies $ \id_n+A^{-1/2}BA^{-1}BA^{-1/2}\geq 0$ and hence
$A+BA^{-1}B\geq 0$ which implies $M\geq 0.$ The case of degenerate
$A$ is obtained by approximation. Thus we have proved

\begin{lemma}
$M\geq 0$ iff $A\pm iB\geq 0$.
\end{lemma}

For a matrix $M$ of the form (\ref{Eq:M}) we denote $\hat{M}=A+iB.$
Then $\hat{ \sigma}=i\id_n$ and the cp-condition
\eqref{eq:XYconstraint} for gauge-covariant channels takes the form
\begin{equation}
\hat{Y}\geq \pm \left( \id_n-\hat{X}\hat{X}^{\ast }\right) .
\label{Eq:pdg}
\end{equation}
Let $\hat{X}^{\ast }=\hat{U}\hat{K}$ be the polar decomposition
of the matrix $\hat{X}^{\ast }$, where $\hat{U}$ is unitary and
$\hat{K}=\sqrt{\hat{X}\hat{X} ^{\ast }}$ is positive. Then the channel $(X,Y)$ is a concatenation of
the reversible channel $(U,0)$ and the channel $(K,Y).$ For the last
channel the condition (\ref{Eq:pdg}) takes the form

\begin{equation}
\hat{Y}\geq \pm \left( \id_n-\hat{K}^{2}\right) .  \label{Eq:pdg1}
\end{equation}

There are several basic cases depending on the properties of
$\hat{K}$:

\begin{itemize}
\item[(i)] $\hat{K}=0.$ The channel is idempotent for any $\hat{Y}\geq \id_n$.
(State preparation with the covariance matrix $Y$).

\item[(ii)] $0<\hat{K}<\id_n,$ where by strict inequality we mean that the
eigenvalues $k_{j}$ of $\hat{K}$ satisfy $0<k_{j}<1.$ The channel is a
member of a one-parameter semigroup of Gaussian channels with bounded
noise term having an invariant state. The semigroup is defined via the
relations
$\hat{X}_{t}=\hat{K}^{t},\hat{Y}_{t}=\hat{\mathcal{Y}}-\hat{K}^{t}
\hat{\mathcal{Y}}\hat{K}^{t},$ where $\hat{\mathcal{Y}}$ is the
unique solution of the equation
\[
\hat{Y}=\mathcal{Y-}\hat{K}\hat{\mathcal{Y}}\hat{K}
\]
corresponding to the covariance matrix of the invariant state.
Indeed, the last equation written in the basis of eigenvectors of
$\hat{K}$ has unique solution
\begin{equation}
\hat{\mathcal{Y}}=\left( \nu _{ij}\right) ,\quad \nu _{ij}=\left(
1-k_{i }k_{j}\right) ^{-1}y_{ij},\quad \hat{Y}=\left( y_{ij}\right)
. \label{Eq:solution}
\end{equation}
Let us prove that $\hat{\mathcal{Y}}\geq \id_n$ and hence it
corresponds to the covariance matrix of a Gaussian state. The matrix
with the elements
\begin{equation}
\left( 1-k_{i}k_{j}\right) ^{-1}=\int_{0}^{\infty }\exp t\left(
k_{i}k_{j}-1\right) dt  \label{Eq:ki}
\end{equation}
is positive since the matrix $\left( k_{i
}k_{j}-1\right) $ is conditionally positive (see Thm. 6.3.6. in \cite{HoJoTopics}). By using the
condition $\hat{Y}\geq \id_n-\hat{K}^{2}$ and a lemma concerning
Schur products of positive matrices (Thm. 5.2.1. in \cite{HoJoTopics}), we have
\[
\hat{\mathcal{Y}}=\left( \left( 1-k_{i}k_{j}\right)
^{-1}y_{ij}\right) \geq \left( \left( 1-k_{i}k_{j}\right)
^{-1}\left( 1-k_{i}^{2}\right) \delta _{ij}\right) =\id_n.
\]

\item[(iii)] $\hat{K}=\id_n.$ Then $\hat{Y}\geq 0$ and the channel is member of
one-parameter semigroup of Gaussian channels $\left( 1,tY\right) $
with unbounded noise term. In the case $\hat{Y}=0$ this is identity
channel.

\item[(iv)] $\hat{K}>\id_n,$ that is $k_{j}>1.$ The channel is a member of
one-parameter semigroup of Gaussian channels with
$\hat{\mathcal{Y}}$ defined as in (\ref{Eq:solution}). However in
this case instead of (\ref{Eq:ki}) we must use
\begin{equation}
\left( 1-k_{i}k_{j}\right) ^{-1}=-\int_{0}^{\infty }\exp t\left(
1-k_{i}k_{j}\right) dt, \label{Eq:kin}
\end{equation}
therefore the matrix with the elements (\ref{Eq:kin}) is negative
implying $ \hat{\mathcal{Y}}\leq -\id_n.$ Thus the semigroup has
unbounded noise term and there is no invariant state.
\end{itemize}

In general, one can decompose $\hat{K}$ into direct orthogonal sum
of the matrices satisfying the conditions (i)-(iv). In case
$\hat{Y}$ commutes with $\hat{K}$ one can further decompose the
channel $(K,Y)$ into corresponding channels. In particular, in the
case $n=1$ any gauge-invariant channel reduces to one of the cases
(i)-(iv).

\section{Conclusions and open questions}\label{sec:open}

We conclude with some open questions. First of all, we lack a
characterization  of infinitesimal divisible Gaussian channels.
 Prop.\ref{prop:weakconv} provides a partial
answer in terms of the determinant $\det(X)$. A similar property,
however, turned out to be necessary but not sufficient in the case
of finite dimensional quantum channels \cite{WoCi08}. Moreover, we left open the question which type of dynamical equations (e.g., time-dependent Markovian master equations) leads to solutions which coincide with the set of infinitesimal divisible channels.

For one-parameter semigroups our picture is more complete. Yet,
there is no efficiently decidable criterion which enables us to say
whether or not a given Gaussian channel is an element of such a
one-parameter semigroup. The simple form \eqref{eq:one-param-simple}
suggests to follow the lines of \cite{WoEiCuCi08,CuEiWo09} where an integer
semi-definite program provided a solution for finite dimensional
quantum channels. However, boundary cases (e.g., channels not
admitting a simple form representation) will have to be treated with
care.

Other questions arise when we slightly change the rules of the game.
In the reversible case we saw that while a transformation might not
be an element of a one-parameter semigroup within the Gaussian
world, it can become one if we drop the restriction to the Gaussian
world. So how does the general picture change if we allow for
decompositions into arbitrary channels?

In a similar vein we may allow for tensor powers and thereby
investigate the  robustness of all the discussed properties w.r.t.
taking several copies of a quantum channel. This might be
interesting beyond Gaussian channels (e.g., for qubit maps) as well.
In the Gaussian case we can easily find examples showing that things
can change: take a reversible Gaussian channel with
$S\notin\exp(\symp)$, then $S\oplus S$ happens to have a Hamiltonian
matrix as a generator.

Finally, it is desirable to relate semigroup properties of a quantum
channel to other properties such as their capacities or to
properties of quantum spin chains to which the channels can be
assigned to via the finitely correlated state construction
\cite{FaNaWe94}, \cite{PeVeWoCi07}.

\section*{Acknowledgements}

T.H. and M.W. acknowledges support by QUANTOP, the Danish
Natural Science Research Council (FNU) and the EU projects QUEVADIS and COQUIT. A.H. acknowledges
support from RFBR grant 09-01-00424 and the RAS program
``Mathematical control theory''.

\end{document}